\documentclass[11pt]{article}

\usepackage[margin=1.3in]{geometry}

\usepackage{amsthm}
\usepackage{amsmath}
\usepackage{natbib}
\usepackage[colorlinks = false,citecolor=blue,urlcolor=blue,filecolor=blue,backref=page]{hyperref}
\usepackage{graphicx}
\usepackage[]{algorithm2e}
\newcommand{\bb}[1]{\boldsymbol{#1}}
\usepackage{authblk}
\usepackage{amsfonts}
\usepackage{comment}

\def\y{\mathbf y}
\def\X{\mathbf X}
\def\x{\mathbf x}
\def\m{\mathbf m}
\def\w{\mathbf w}
\def\V{\mathbf V}

\def\EE{\mathbb{E}}

\def\bb{\boldsymbol{\beta}}
\def\gg{\boldsymbol{\gamma}}
\def\sfN{\textsf{N}}
\def\sfBern{\textsf{Bern}}
\def\sfBeta{\textsf{Beta}}
\def\sfIG{\textsf{IG}}
\def\Var{\text{Var}}
\def\diag{\text{diag}}
\def\e{\mathbf e}
\def\b{\mathbf b}
\newtheorem{thm}{Theorem}[section]
\begin{document}

\title{An Ensemble EM Algorithm for Bayesian Variable Selection}

\author[1]{Jin Wang \thanks{jinwang8@illinois.edu}}
\author[1]{Feng Liang \thanks{liangf@illinois.edu}}
\author[2]{Yuan Ji \thanks{jiyuan@uchicago.edu}}
\affil[1]{ Department of Statistics,
 University of Illinois at Urbana-Champaign}
\affil[2]{Department of Biostatistics, University of Chicago}

\maketitle

\begin{abstract}
We study the Bayesian approach to variable selection in the context of linear regression. Motivated by a recent work by \cite{George:2014}, we propose an EM algorithm that returns the MAP estimate of the set of relevant variables. Due to its particular updating scheme,  our algorithm can be implemented efficiently without 
inverting a large matrix in each iteration and therefore can scale up with big data. We also show that the MAP estimate returned by our EM algorithm achieves variable selection consistency even when $p$ diverges with $n$. In practice, our algorithm could get stuck with local modes, a common problem with EM algorithms. To address this issue, we propose an ensemble EM algorithm, in which we repeatedly apply the EM
algorithm on a subset of the samples with a subset of the covariates, and
then aggregate the variable selection results across those bootstrap
replicates. Empirical studies have demonstrated the superior performance of 
the ensemble EM algorithm. 
\end{abstract}


\section{Introduction}

Consider a simple linear regression model with Gaussian noise:
\begin{equation} \label{linearmodel}
 \y =  \X \bb + \e
 \end{equation}
where $\y=(y_1, \dots, y_n)^T$ is the $n \times 1$ response, $\e=(e_1, \dots, e_n)^T$ is a vector of iid Gaussian random variables with mean $0$ and variance $\sigma^2$, and $\X$ is the $n \times p$ design matrix. The unknown parameters are the regression parameter $\bb = (\beta_1, \dots,
\beta_p)^T$ and the error variance $\sigma^2$.  In many real applications
such as bioinformatics and image analysis, where linear regression models
have been routinely used, the number of potential predictors (i.e., $p$)
is large but only a small fraction of them is believed to be relevant. Therefore the linear model \eqref{linearmodel} is often assumed to be ``sparse" in the sense that most of the coefficients $\beta_j$'s are zero. Estimating the set of relevant variables, $S = \{j : \beta_j \ne 0 \}$,  is an important problem in modern statistical analysis. 

The Bayesian approach to variable selection is conceptually simple and straightforward. First introduce a  $p$-dimensional binary vector $\gg = (\gamma_1, \dots, \gamma_p)^T$ to index all the $2^p$ sub-models, where $\gamma_j=1$ if the $j$th variable is included in this model and $0$ if excluded. Usually $\gamma_j$'s are modeled by independent Bernoulli distributions. Given $\gg$, a popular prior choice for $\bb$ is the ``spike and slab" prior 
\citep{Mitchell:1988}:
\begin{equation} \label{eq:spike:slab}
\pi(\beta_j \mid \gamma_j) =
\begin{cases}
 \delta_0(\beta_j), & \text{if } \gamma_j=0;
\\
g(\beta_j), & \text{if } \gamma_j=1,
\end{cases}
\end{equation}
where $\delta_0(\cdot)$ is the Kronecker delta function corresponding to the density function of a point mass at $0$ and $g$ is  a continuous density function. After specifying priors on all the unknowns, one needs to calculate the posterior distribution. 
Most algorithms for Bayesian variable selection rely on MCMC such as Gibbs or Metropolis Hasting to obtain the posterior distribution; for a review on recent developments in this area, see \citet{o2009review}.  MCMC algorithms, however, are insufficient to meet the growing demand on scalability from real applications. Since the primary goal is variable selection, we focus on efficient algorithms that return the MAP estimate of $\gg$, as an alternative to these MCMC-based sampling methods that return the whole posterior distribution on all the unknown parameters. 

Recently, \citet{George:2014} proposed a simple, elegant EM algorithm for Bayesian variable selection. They  adopted a continuous version of the  ``spike and slab" prior---the spike component in \eqref{eq:spike:slab} is replaced by a normal distribution with a small variance \citep{George:1993}, and proposed an EM algorithm to obtain the MAP estimate of the regression coefficient $\bb$. The MAP estimate $\hat{\bb}_{\text{MAP}}$, however, is not sparse, and an additional thresholding step is needed to estimate $\gg$.

In this paper, we develop an EM algorithm that directly returns the MAP estimate of $\gg$, so no further thresholding is needed. We adopt the same continuous ``spike and slab" prior. Different from the algorithm by \citet{George:2014} that returns $\hat{\bb}_{\text{MAP}}$ by treating $\gg$ as latent, our algorithm returns the MAP estimate of the model index, $\hat{\gg}_{\text{MAP}}$, by treating $\bb$ as latent. The special structure of our EM algorithm allows us to use a computational trick to avoid inverting a big matrix at each iteration, which seems unavoidable in the algorithm by \citet{George:2014}. Further we can show that the $\hat{\gg}_{\text{MAP}}$ achieves asymptotic consistency even when $p$ diverges to infinity with the sample size $n$.


Although shown to achieve selection consistency, in practice, our EM algorithm could get stuck at a local mode due to the large discrete space in which $\gg$ lies. Borrowing the idea of bagging, we propose an ensemble version of our EM algorithm (which we call BBEM): apply the algorithm on multiple Bayesian bootstrap (BB) copies of the data, and then aggregate the variable selection results.  Bayesian bootstrap for variable selection was explored before by \cite{Clyde:2001} for the purpose of prediction, where models built on different bootstrap copies are combined to predict the response. But the focus of our approach is to summarize the evidence for variable relevance from multiple BB copies, which is similar in nature to several frequentist ensemble methods for variable selection, such as the AIC ensemble \citep{Zhu:2006}, stability selection \citep{Meinshausen:2010}, and  random Lasso \citep{Wang:2011}.

The remaining of the paper is organized as follows. Section 2 describes the EM algorithm in detail, Section 3 presents the asymptotic results, and Section 4 describes the BBEM algorithm. Empirical studies are presented in Section 5 and  conclusions and  remarks in Section 6. 


\section{The EM Algorithm}\label{sec:em}
\subsection{Prior Specification}

We adopt the continuous version of ``spike and slab" prior for $\bb$, i.e. a mixture of two normal components with mean zero and different variances: 
\begin{equation} \label{cont:spike:slab}
\pi(\beta_j \mid \sigma, \gamma_j) =
\begin{cases}
\sfN(0,\sigma^2 v_0), & \text{if } \gamma_j=0;
\\
\sfN(0,\sigma^2 v_1), & \text{if } \gamma_j=1,
\end{cases}
\end{equation}
where $v_1 > v_0 > 0.$ Alternatively, we can write the prior on $\bb$ as 
$$ 
\pi(\beta_j \mid \sigma^2, \gamma_j) = \sfN(0,\sigma^2 d_{\gamma_j}),$$
where 
$$ d_{\gamma_j} =  \gamma_j v_1 + (1 - \gamma_j) v_0. $$

For the remaining parameters, we specify independent Bernoulli priors on elements of $\gg$, and conjugate priors like Beta and Inverse Gamma on $\theta$ and $\sigma^2$, respectively: 
\begin{eqnarray*}
\pi(\gg\mid\theta) &=& \sfBern(\theta), \\
\pi(\theta) &=& \sfBeta(a_0,b_0), \\
 \pi(\sigma^2) &=& \sfIG(\nu/2, \nu\lambda/2).
\end{eqnarray*}
For hyper-parameters $(a_0, b_0, \nu, \lambda)$, we suggest the following non-informative choices unless prior knowledge is available: 
\begin{equation} \label{value:hyper:para}
 a_0=b_0=1.1, \quad \nu=\lambda=1. 
\end{equation}
The choice for $v_0$ and $v_1$ will be discussed later. 

\subsection{The Algorithm}\label{sec:alg}
With the Gaussian model and prior distributions specified above, we can write down the full posterior distribution: 
\begin{equation*}
\pi(\gg, \bb, \theta, \sigma^2 \mid \y)\ \propto\  p(\y \mid \bb, \sigma^2) \times \pi(\bb \mid \sigma, \gg) \times \pi(\gg\mid\theta) \times  \pi(\theta) \times \pi(\sigma^2).
\end{equation*}
Treating $\bb$ as the latent variable, we derive an EM algorithm that returns the MAP estimation of parameters $\Theta = (\gg, \sigma^2, \theta)$, whereas  the roles of $\bb$ and $\gg$ are switched in \citet{George:2014}. 

\subsubsection*{E Step}

The objective function $Q$ at the $(t+1)$-th iteration in an EM algorithm is defined as the integrated logarithm of the full posterior with respect to $\bb$ given $\y$ and the parameter values from the previous iteration $\Theta^{(t)} = (\gg^{(t)}, \sigma_{(t)}^2, \theta^{(t)})$, i.e.,
\begin{eqnarray}
Q(\Theta \mid \Theta^{(t)}) &=& \EE_{\bb|\Theta^{(t)}, \y}\log \pi(\Theta, \bb \mid \y) \nonumber\\
&=& - \frac{1}{2 \sigma^2}  \EE_{\bb|\Theta^{(t)}, \y} \Big [  \|\y-\X\bb\|^2 + 
  \sum_{j=1}^p \frac{\beta_j^2}{d_{\gamma_j}}  \Big ] + F(\Theta), \label{Q:ex}  
\end{eqnarray}
where 
\begin{eqnarray*}
F(\Theta) &=& - \frac{n+p}{2} \log \sigma^2 - \frac{1}{2} \sum_{j=1}^p \log d_{\gamma_j} +   \pi(\gg\mid\theta)  \\
&& + \log  \pi(\theta) + \log  \pi(\sigma^2) + \text{Constant}
\end{eqnarray*}
 is a function of $\Theta$ not depending on $\bb$. 

It is easy to show that $\bb$ follows a Normal distribution with mean $\m$ and covariance matrix $\sigma_{(t)}^2 \V$, given $\Theta^{(t)}$ and  $\y$, where
\begin{eqnarray}
\m &=&  \V^{-1} \X^T\y, \quad \V = \big (\X^T \X +D_{\gg^{(t)}}^{-1} \big )^{-1},  \label{eq:mD} \\
D_{\gg^{(t)}} &=& \diag \Big ( d_{\gamma^{(t)}_j} \Big )_{j=1}^p = \diag \Big ( \gamma^{(t)}_j v_1 + (1 - \gamma^{(t)}_j) v_0 \Big )_{j=1}^p.\nonumber
\end{eqnarray}
Then the two expectation terms in \eqref{Q:ex} can be expressed as: 
\begin{eqnarray}
\EE_{\bb|\Theta^{(t)}, \y} \ \big \|\y-\X\bb \big \|^2
&=& \sigma_{(t)}^2  \text{tr}(\X \V \X^T) + \big \|\y-\X\m \big \|^2, \label{eq:residual}
\\
\EE_{\bb|\Theta^{(t)}, \y} \sum_{j=1}^p \frac{\beta_j^2}{d_{\gamma_j}} &=& \sum_{j=1}^p \frac{\sigma_{(t)}^2 V_{jj} + m_j^2 }{ (1-\gamma^{(t)}_{j})v_0+\gamma^{(t)}_{j} v_1}.
\label{eq:exp2}
\end{eqnarray}

\subsubsection*{M Step}
We sequentially update parameters $(\gg,\theta,\sigma)$ to maximize the objective function $Q$. 

\begin{enumerate}
\item {\bf Update $\gamma_j$'s. } The terms involving $\gamma_j$ in (\ref{Q:ex}) are
\begin{equation}
- \frac{1}{2 \sigma^{2}_{(t)}}  \EE_{\bb|\Theta^{(t)}, \y} \left [ \frac{\beta_j^2}{d_{\gamma_j}}  \right ] - \frac{1}{2} \log d_{\gamma_j} + \log \pi(\gamma_j \mid \theta^{(t)}).
\label{eq:Qgamma}
\end{equation}
Plug in $\gamma_j=0$ and $\gamma_j=1$ to (\ref{eq:Qgamma}) respectively, then we have 
\begin{equation}
\gamma_j^{(t+1)}  =  1, \quad\text{if}\quad  \EE_{\bb|\Theta^{(t)}, \y} \big [ \beta_j^2 \big ] > r^{(t)} ,
\label{eq:gam1}
\end{equation}
where \[ r^{(t)} = \frac{\sigma^{2}_{(t)}}{1/v_0 - 1/v_1} \Big  (\log{v_1\over v_0} -2 \log\frac{\theta^{(t)}}{1-\theta^{(t)}} \Big ). \]

\item {\bf Update $(\sigma^2, \theta).$} Given $\gg^{(t+1)}$, the updating equations for the other two parameters are given by
\begin{equation}
\sigma^{2}_{(t+1)} = \frac{ \EE_{\bb|\Theta^{(t)}, \y} \Big [  \|\y-\X\bb\|^2 + 
  \sum_{j=1}^p \beta_j^2 / d_{\gamma^{(t+1)}_j}  \Big ]+  \nu \lambda }{n+p+\nu}, \label{eq:sig1}
\end{equation}
\begin{equation}
\theta^{(t+1)} = \frac{\sum^p_{j=1}\gamma_j^{(t+1)} +a_0-1}{p+a_0+b_0-2}. \label{eq:the1}
\end{equation}
\end{enumerate}

\subsubsection*{Stopping Rule}

The EM algorithm alternates between the E-step and M-step until convergence. A natural stopping criterion is to check whether the change of the objective function $Q$ is small.  
To reduce the computation cost for evaluating the $Q$ function, we adopt a different stopping rule as our main focus is $\gg$: we stop our algorithm when the estimate $\gg^{(t)}$ stays the same for $k_0$ iterations. In practice, we suggest to set $k_0=3.$ The pseudo code of this EM algorithm is summarized in Algorithm \ref{arg:em}.

\RestyleAlgo{boxruled}
\begin{algorithm}
 \caption{EM Algorithm}
 \label{arg:em}
 \KwIn{$\X, \y, v_0, v_1,a_0,b_0,\nu,\lambda $}
 Initialize $\Theta^{(0)}$\;
 E-step: Calculate the two expectations in (\ref{eq:residual}) and (\ref{eq:exp2}), denoted as $EE^{(0)}$\;
 \For{t = 1 : maxIter}{
  M-step: Update $\Theta^{(t)}$ from Eq (\ref{eq:gam1},~\ref{eq:sig1},~\ref{eq:the1})\;
  E-step: Update $EE^{(t)}$ from Eq (\ref{eq:residual},~\ref{eq:exp2})\;
  \If{$\gg^{(t)}$ stays the same for $k_0=3$ iterations}{
   break\;
   }
}
 Return $\gg$, $\m$\;
\end{algorithm}

\subsection{Computation Cost}\label{sec:fast}
At each E-step, updating the posterior of $\bb$ given other parameters in (\ref{eq:mD}) requires inverting a $p\times p$ matrix 
\begin{equation} \label{eq:V}
\V_{(t)} = (\X^T \X +D_{\gg^{(t)}}^{-1})^{-1}, 
\end{equation} 
which is the major computational burden of this algorithm. 
When $p>n$, we can use the Sherman-Morrison-Woodbury formula to compute the inverse of an $n\times n$ matrix. So the computation cost at each iteration is of order $O(\min(n,p)^3)$. It is, however, still time-consuming when both $n$ and $p$ are large. 

Note that the only thing that changes in (\ref{eq:V}) from iteration to iteration is $D_{\gg^{(t)}}$,  a diagonal matrix depending on the binary vector $\gg^{(t)}$. From our experience, only a small fraction of $\gamma^{(t)}_j$'s are changed at each iteration after the first a couple of iterations. So the idea is to use the following recursive formula to compute $\V_{(t)}$: 
\begin{eqnarray}
\V_{(t)} &=& (\X^T \X + D^{-1}_{\gg^{(t-1)}} + D^{-1}_{\gg^{(t)}} - D^{-1}_{\gg^{(t-1)}} )^{-1} \nonumber \\
&=& (\V_{(t-1)}^{-1} + D^{-1}_{\gg^{(t)}} - D^{-1}_{\gg^{(t-1)}})^{-1} \label{eq:trick}
\end{eqnarray}
where $D^{-1}_{\gg^{(t)}} - D^{-1}_{\gg^{(t-1)}}$ is a diagonal matrix with the $j$-th diagonal entry being non-zero only if the inclusion/exclusion status, i.e., the value of $\gamma_j$, is changed from the last iteration. Let $l$ denote the number of variables whose  $\gamma_j$ values are changed from iteration $(t-1)$ to $t$. Then $D^{-1}_{\gg^{(t)}} - D^{-1}_{\gg^{(t-1)}}$ is a rank $l$ matrix. We can apply the Woodbury formula on (\ref{eq:trick}) to  reduce the computation complexity from $O(\min(n,p)^3)$ to $O(l^3)$.

For example, without loss of generality, suppose only the first $l$ covariates have their $\gamma_j$ values changed. Then, we can write 
$$D^{-1}_{\gg^{(t)}} - D^{-1}_{\gg^{(t-1)}}= U_{p\times l}A_{l\times l}U^T, $$
where $A = \big ( \frac{1}{v_0}-\frac{1}{v_1} \big ) \diag(2\gamma_j^{(t)}-1)_{j=1}^l$ and $U$ consists of the first $l$ columns from $\mathbf{I}_p.$ Applying the Woodbury formula, we have
$$
\V_{(t)} 
=\V_{(t-1)}-\V_{(t-1)} U (A^{-1}+U^T \V_{(t-1)} U)^{-1} U^T\V_{(t-1)}. 
$$


\section{Asympototic Consistency}\label{sec:asym}

In this section, we study the asymptotic property of $\hat{\gg}_n$, the MAP estimate of model index returned by our EM algorithm. Assume the data $\y_n$ are generated from a Gaussian regression model: 
$$\y_n \sim \textsf{N}_n \big ( \X_n \bb^\ast_n, \sigma^2 \mathbf{I}_n \big ).$$
Here we consider a triangular array set up: the dimension $p=p_n$ diverges with $n$ and the true coefficients $\bb^\ast_n$ also vary with $n$. 
Suppose the true model is indexed by $\gg_n^\ast$, where $\gamma^\ast_{nj}=1$ if $\beta^\ast_{nj} \ne 0$ and  $\gamma^\ast_{nj}=0$ if $\beta^\ast_{nj} = 0$. We show that  our EM algorithm has the following selection consistency property:
$$\mathbb{P}( \hat{\gg}_n = \gg_n^*) \to 1,\quad \text{as } n \to \infty. $$

First we list some regularity conditions needed in our proof. Let $\lambda_{\text{min}}(A)$ denote the smallest eigenvalue of matrix $A$. We assume
\begin{align*}
(A1) \quad & \lambda_{\min}(\X_n^T\X_n)^{-1}=O(n^{-\eta_1}),\ 0<\eta_1\le 1; \\
(A2) \quad & \|\bb_n^\ast\|_2 =O(n^{\eta_2}),\ 0<\eta_2<\eta_1; \\
(A3) \quad & \liminf_n \frac{\min \big \{ |\beta_{nj}^\ast|, \gamma_{nj}^\ast=1 \big \} }{n^{(\eta_3-1)/2}}\ge M,\ 0\le \eta_3<1; \\
(A4) \quad & a_0\sim p_n,\ b_0\sim p_n,\ \nu=\infty,\ \lambda=1,
\end{align*}
where $M$ is a positive constant, and $(a_0, b_0, \nu, \lambda)$ are the hyper-parameters from the Beta and InvGamma priors. 

Assumption (A1) controls the collinearity among covariates; in the traditional asymptotic setting where $p$ is fixed, we have $\eta_1=1.$ Assumption (A2) controls the sparsity (in terms of $L_2$ norm) of the true regression coefficient vector.  
Assumption (A3) requires that the minimal non-zero coefficient cannot go to zero at a rate faster than $1/\sqrt{n}$; in the traditional asymptotic setting where $\bb^\ast$ is fixed, we have $\eta_3 = 0.$ Assumption (A4) is purely technical, which ensures that $\hat{\theta}_n$ and $\hat{\sigma}_n^2$ are bounded. In fact we could fix $\hat{\theta}_n$ and $\hat{\sigma}_n^2$ to be any constant, which does not affect the proof. In our simulation studies, we still recommend \eqref{value:hyper:para} as the choice for hyper-parameters unless $p$ is large. 

\begin{thm}\label{thm:consistency}
Assume (A1-A4) and $p=O(n^\alpha)$ where $0\le \alpha <1$. With $v_1$ fixed and $v_0$ satisfying $$0<v_0=O(n^{-r_0}),  \quad 1-\eta_3 < r_0 < \min \Big \{\eta_1-\alpha, \frac{2}{3}(\eta_1-\eta_2) \Big \},$$
the model returned by our EM  algorithm, $\hat{\gamma}_n$,  achieves the following selection consistency,
\begin{equation}
\mathbb{P}( \hat{\gg}_n = \gg_n^*) \to 1,\quad \text{as } n \to \infty.
\end{equation}
\end{thm}

\begin{proof}
See Appendix. 
\end{proof}

\section{The BBEM Algorithm}\label{sec:bb}
A common issue with EM algorithms is that they could be trapped at a local maximum. There are some standard remedies available for dealing with this issue, for instance, trying a set of different initial values or utilizing some more advanced  optimization procedures at the M-step. Since our EM algorithm is searching for the optimal $\gg$ over a big discrete space, all $p$-dimensional binary vectors, these remedies are less useful when $p$ is large. 

When doing optimization with $\gg$, a discrete vector, the resulting solution is often not stable, i.e., has a large variance. Bagging is an easy but powerful method \citep{Breiman:1996} for variance reduction, which applies the same algorithm on multiple bootstrap copies of the data, and then aggregates the results. We proposed the following ensemble EM algorithm, in which we repeatedly run the EM variable selection algorithm, Algorithm \ref{arg:em} from Section \ref{sec:alg}, on Bayesian bootstrap replicates. 

The original bootstrap repeatedly draws samples  from the original data set $\{ (\x_i, y_i )\}_{i=1}^n$
 with replacement, i.e., each observation $(\x_i, y_i)$ is sampled with probability $1/n.$ In Bayesian bootstrap \citep{Rubin:1981}, instead of sampling a subset of the data, we assign a random weight $w_i$ to the $i$-th observation  and then fit a weighted least squares regression model on the whole data set. In particular, following \cite{Rubin:1981}, we generate the  weights $\w = (w_1, \dots, w_n)$ from a n-category Dirichlet distribution: 
\begin{equation} \label{eq:bsweight}
\w_{n\times 1}\sim \sf{Dir}(1,\cdots,1). 
\end{equation}
When applying Algorithm \ref{arg:em} on a weighted linear regression model, all the updating equations stay the same,  except equation  (\ref{eq:mD})  for the posterior of $\bb$, which should be  changed to:
\begin{equation}
\m = \V \X^T \diag(\w) \y,\quad
\V = (\X^T \diag(\w) \X +D_{\gg^{(t)}}^{-1})^{-1}. 
\label{eq:postbetabs}
\end{equation}
Eq (\ref{eq:residual}), the expectation of the weighted residual sum of squares, should also be changed accordingly: 
\begin{equation}
\EE_{\bb|\Theta^{(t)}, \y} \ \big  \|\y-\X\bb \big \|^2_\w = \sigma_{(t)}^2 \text{tr}(\diag(\w)\X \V \X^T) +  ( \y-\X\m )^T \diag(\w)  ( \y-\X\m ). \label{eq:exp1}
\end{equation}


It is well-known  that in order to make the aggregation work, we should control the correlation among estimates from bootstrap replicates. For example, in  random forest \citep{Breiman:2001},   the number of variables used for choosing the optimal split of a tree is restricted to a subset of the variables, instead of using all $p$ variables. A similar idea was implemented in Random Lasso \citep{Wang:2011}, an  ensemble algorithm for variable selection. In the same spirit, we apply the EM algorithm only on a subset of the variables at each Bayesian bootstrap iteration. A naive way is to randomly pick a subset from the $p$ variables. This, however, will be inefficient when $p$ is large and the true model is sparse, since it is likely most random subsets will not contain any relevant variables. So we employ a biased  sampling procedure:  sample the $p$ variables based on a weight vector $\tilde{\pi}$ that is defined as
\begin{equation}\label{eq:varweight}
\tilde{\pi}_{p\times1}\propto |\X^T \y |/\diag(\X^T \X),
\end{equation}
that is, variables are sampled based on their marginal effect in a simple linear regression. 

The ensemble EM algorithm operates as follows. First we sample a random set of $L$ variables according to the probability vector $\tilde{\pi}$, and draw a $n\times 1$ bootstrap weight vector $\w$ from (\ref{eq:bsweight}).  Let $\tilde{\X}$ be the new data matrix with the $L$ columns. Then apply the EM algorithm on  $\tilde{\X}$ with weight $\w$. 
 Let $\gg_k$ denote the model returned by the $k$-th Bayesian bootstrap iteration, where the $j$-th element of $\gg_k$ is $1$ if the $j$-th variable is selected and zero otherwise; of course, the $j$-th element is zero if the $j$-th variable is not included in the initial $L$ variables. Define the final variable selection frequency for the $p$ variables as
\begin{equation} \label{eq:phi:BBEM}
\boldsymbol{\phi}_{p \times 1} = \frac{1}{K} \sum_{k=1}^K \gg_k. 
\end{equation}
We can report the final variable selection result by thresholding $\phi_j$'s at some fixed number, for example, a half. Or we can produce a path-plot of $\boldsymbol{\phi}$ as $v_0$ varies, which could be a useful tool to investigate the importance of each variable. We illustrate this in our simulation study in Section \ref{sec:simu}. 


As for the computational cost, the inversion of the $L\times L$ matrix in \eqref{eq:postbetabs} is a big improvement compared with that of a $p\times p$ matrix, while it can be further simplified through the fast computing trick in Section \ref{sec:fast}. We call this algorithm, BBEM, which  is summarized in Algorithm \ref{arg:bbem}.

\RestyleAlgo{boxruled}
\begin{algorithm}
 \caption{BBEM Algorithm}
 \label{arg:bbem}
 \KwIn{$\X, \y, v_0, v_1,a_0,b_0,\nu,\lambda,K,L $}
 Compute the variable weight $\tilde{\pi}$ from (\ref{eq:varweight})\;
 \For{k = 1 : K}{
 Generate a subset of $L$ variables according to $\tilde{\pi}$\;
 Make the replicate $\tilde{\X}^k$ with the $L$ variables\;
 Initialize $\Theta^{(0)}_k$\;
 Generate bootstrap weight $\w$ from (\ref{eq:bsweight})\;
 E-step: Calculate the two expectations in (\ref{eq:exp2}), denoted as $EE^{(0)}_k$\;
 \For{t = 1 : maxIter}{
  M-step: Update $\Theta^{(t)}_k$ from Eq (\ref{eq:gam1},~\ref{eq:sig1},~\ref{eq:the1})\;
  E-step: Update $EE^{(t)}_k$ from Eq (\ref{eq:exp1},~\ref{eq:exp2})\;
  \If{$\gg^{(t)}_k$ stays the same for $k_0=3$ iterations}{
   break\;
   }
   }
   Record $\gg^{(t)}_k$, $\m^{(t)}_k$\;   
 }
 Return $\boldsymbol{\phi}$ from Eq (\ref{eq:phi:BBEM})\;
\end{algorithm}

\section{Empirical Study}\label{sec:simu}
In this section, we first compare the proposed EM algorithm (Algorithm \ref{arg:em}) with other popular methods on a widely used benchmark data set. Then we compare BBEM (Algorithm \ref{arg:bbem}) with other methods on two more challenging data sets of larger dimensions. Finally, we applied  BBEM  on a restaurant revenue data from a Kaggle competition, and showed that our algorithm outperforms the benchmark from random forest. 

For the hyper-parameters $v_0$ and $v_1$, we set $v_1=100$ as fixed and  tune  an appropriate value for $v_0$ either based on 5-fold cross-validation or BIC. 
For the initial value for $\theta$, we suggest to use $1/2$ for ordinary problems, but $\sqrt{n}/p$ for large-$p$ problems. The initial value of $\sigma^2$ is set as $1$. In addition, there are two bootstrap parameters: the total number of replicates $K$ and the number of variables used in each bootstrap $L$. For efficiency, the number of variables in each bootstrap replicate should not exceed the sample size $n$. We use $K=100$, and  $L=n/2=50$ if $p$ is large and $L=p$ is $p$ is small.

\subsection{A widely used benchmark} 
First we apply our EM algorithm on a widely used benchmark data set \citep{Tibshirani:1996}, which has $p=8$ variables, each from a standard normal distribution with pairwise correlation $\rho(\x_i,\x_j)=0.5^{|i-j|}$. The response variable is generated from  
$$\y = 3\x_1 + 1.5\x_2 +2\x_5 + \epsilon$$
where $\epsilon\sim \sfN(0,\sigma^2)$. 

Following \cite{Fan:2001}, we repeat the experiment 100 times under two scenarios: 
(1) $n=40, \sigma=3$ and (2) $n=60, \sigma=1$. The result is shown in Table \ref{table:benchmark}, which reports the average number of zero-coefficients (i.e., no selection) among signal  variables ($\x_1,\x_2,\x_5$) and  among noise variables, respectively. The results for SCAD1 (tuning parameter selected by cross-validation), SCAD2 (tuning parameter fixed) and LASSO are taken from \cite{Fan:2001}. In the first ``small sample-size high noise" scenario, our EM algorithm has the highest number of  zero-coefficients among noise variables, i.e., the lowest type I error. The average number of signal variables missed by EM is slightly higher than SCAD1 (where the tuning parameter is chosen by cross-validation) but less than SCAD2 (where the tuning parameter is pre-fixed). But overall, our EM algorithm and the two SCAD methods perform the best. 
In the second ``large sample-size low noise" scenario, no signal variables are missed by any method, but EM has the lowest type I error.
 
\begin{table}[h]
\centering
\begin{tabular}{p{25mm}|p{20mm}p{20mm}}
\hline
        Method   & $\x_j\in$ Noise\newline ($j$=3,4,6,7,8) & $\x_j\in$ Signal\newline ($j$=1,2,5) \\
           \hline
  $n=40$, $\sigma=3$ & &\\
  EM   &   4.55          &  0.24 \\         
SCAD1      &   4.20          &  0.21 \\
SCAD2      &   4.31          &  0.27 \\
LASSO      &   3.53          &  0.07 \\
Oracle     &   5.00          &  0.00 \\ 
\hline
$n=60$, $\sigma=1$ & &\\
EM     &   4.72          &  0.00 \\      
SCAD1      &   4.37          &  0.00 \\
SCAD2      &   4.42          &  0.00 \\
LASSO      &   3.56          &  0.00 \\
Oracle     &   5.00          &  0.00 \\ 
\hline
\end{tabular}
\caption{A widely used benchmark. The average number of zero-coefficients (i.e., no selection) out of 100 simulations for each types of variable (Signal or Noise) are shown.  The results other than EM  (Alg \ref{arg:em}) are  from \cite{Fan:2001}.}
\label{table:benchmark}
\end{table}

Following \cite{Wang:2011} and \cite{Zhu:2012}, we repeat the experiment 100 times with the same sample size $n=50$ but two different noise levels: low noise level ($\sigma=3$) and high noise level ($\sigma=6$). Table \ref{table:benchmark2} reports the minimum, median, maximum of being selected out of 100 simulations for the signal and the noise variables, respectively. Both Lasso and random Lasso have a higher chance of selecting the signal variables, but at the price of mistakenly including many noise variables. Overall, our EM algorithm performs the best, along with PGA and stability selection, two frequentist ensemble methods for variable selection. 

\begin{table}[h]
\centering
\begin{tabular}{l|rrr|rrr}
\hline
       Method   & \multicolumn{3}{c}{$\x_j\in$ Signal  ($j$=1,2,5)} & \multicolumn{3}{c}{$\x_j\in$ Noise ($j$=3,4,6,7,8) } \\
   &     Min &  Median &  Max &  Min & Median &   Max \\
\hline
  $n=50$, $\sigma=3$ & & & & & & \\         
  EM  &   91 & 97 & 100  & 3 & 6 & 12 \\
Lasso          &   99 & 100 & 100  & 48 & 55 & 61 \\
Random Lasso &   95 & 99 & 100   &  33 & 40 & 48 \\
ST2E      &   89 & 96 & 100   & 4 & 12 & 20 \\
PGA      &   82 & 98 & 100   &  4 & 7 & 11 \\
Stability selection & & & & & &\\
\quad$\lambda_{min}=1$    &   81 & 83 & 100  & 0 & 2 & 9 \\
\quad$\lambda_{min}=0.5$  &   90 & 98 & 100  & 4 & 8 & 22 \\
\hline
  $n=50$, $\sigma=6$ & & & & & &\\    
  EM  &   53 & 67 & 91  & 6 & 10 & 14 \\     
Lasso          &   76 & 85 & 99   & 47 & 49 & 53 \\
Random Lasso &   92 & 94 & 100   & 40 & 48 & 58 \\
ST2E      &   68 & 69 & 96   & 9 & 13 & 21 \\
PGA      &    54 & 76 & 94   & 9 & 14 & 16 \\
Stability selection & & & & & &\\
\quad$\lambda_{min}=1$  &   59 & 61 & 92   & 4 & 8 & 18 \\
\quad$\lambda_{min}=0.5$ &  76 & 84 & 100 & 30 & 42 & 50 \\
\hline
\end{tabular}
\caption{A widely used benchmark. The min, median, max number of being selected out of 100 simulations for each types of variable (Signal or Noise) are shown.  The results other than EM  (Alg \ref{arg:em}) are  from \cite{Zhu:2012}.}
\label{table:benchmark2}
\end{table}

\subsection{A highly-correlated data} 
Next we demonstrate our two algorithms on a highly-correlated example from \cite{Wang:2011}. The data has $p=40$ variables and the response $\y$ is generated from 
$$ \y = 3\x_1+3\x_2-2\x_3+3\x_4+3\x_5-2\x_6+\epsilon,$$
where $\epsilon\sim\sfN(0,\sigma^2)$ and $\sigma=6$. Each $\x_i$ is generated from a standard normal with the following correlation structure among the first six signal variables: the signal variables are divided into two groups, $V_1=\{\x_1,\x_2,\x_3\}$ and $\V_2=\{\x_4,\x_5,\x_6\}$; the within group correlation is $0.9$ and the between-group correlation is $0$. 

We repeat the simulation 100 times with $n=50$ and $n=100$, and the 
results are summarized  in Table \ref{table:corr}. 
For this example, due to the high correlation among features we expect ensemble methods to perform better. Indeed, BBEM has the best performance in terms of selecting true signal variables while controlling the error of including noise variables. The performance of the EM algorithm, although not the best, is also comparable with other top ensemble methods like random Lasso from \cite{Wang:2011}, and T2E and PGA from \cite{Zhu:2012}.

\begin{table}[h]
\centering
\begin{tabular}{l|rrr|rrr}
\hline
       Method   & \multicolumn{3}{c}{$\x_j\in$ Signal ($j$= 1:6)} & \multicolumn{3}{c}{$\x_j\in$ Noise  }\\
   &     Min &  Median &  Max &  Min & Median &   Max \\
\hline
  $n=50$, $\sigma=6$ & & & & & & \\         
Lasso          &   11 & 70 & 77  & 12 & 17 & 25 \\
Random Lasso &   84 & 96 & 97   &  11 & 21 & 30 \\
ST2E      &   85 & 96 & 100   & 18 & 25 & 34 \\
PGA      &   55 & 87 & 90   &  14 & 23 & 32 \\
EM         &   65 & 85.5 & 89  & 4 & 10 & 13 \\
BBEM &   89 & 96 & 100   & 4 & 8 & 15 \\
\hline
  $n=100$, $\sigma=6$ & & & & & &\\         
Lasso          &   8  & 84 & 88 & 12 & 22 & 31 \\
Random Lasso &   89  & 99  & 99  & 8  & 14  & 21 \\
ST2E      &   93  & 100  & 100  & 14  & 21  & 27 \\
PGA      &    40  & 85  & 92  & 13  & 22  & 33 \\
EM          &   84 & 91 & 95  & 1 & 7 & 16 \\
BBEM  &   95 & 99 & 100   & 4 & 9 & 14 \\

\hline
\end{tabular}
\caption{A highly-correlated data. The min, median, max number of times being selected (i.e., no selection) out of 100 simulations for each type of variables (Signal and Noise) are shown. The results other than EM and BBEM are  from \cite{Zhu:2012}.}
\label{table:corr}
\end{table}

For illustration purpose, we apply BBEM on a data set with $n=50$ and $v_0$ varying from $10^{-4}$ to $1$. Figure~\ref{fig:tuningcorrelated} shows the path-plot of the selection frequency from BBEM. There is clearly a gap between the signal variables and the noise ones. For a range of $v_0$,  from $0.001$ to $0.02$, BBEM can successfully select the six true variables $\{\x_1,\x_2,\ldots,\x_6\}$ if we threshold the selection frequency $\phi_j$ at $0.5$. 

\begin{figure}[ht] 
 \centering 
 \scalebox{0.65} 
 {\includegraphics{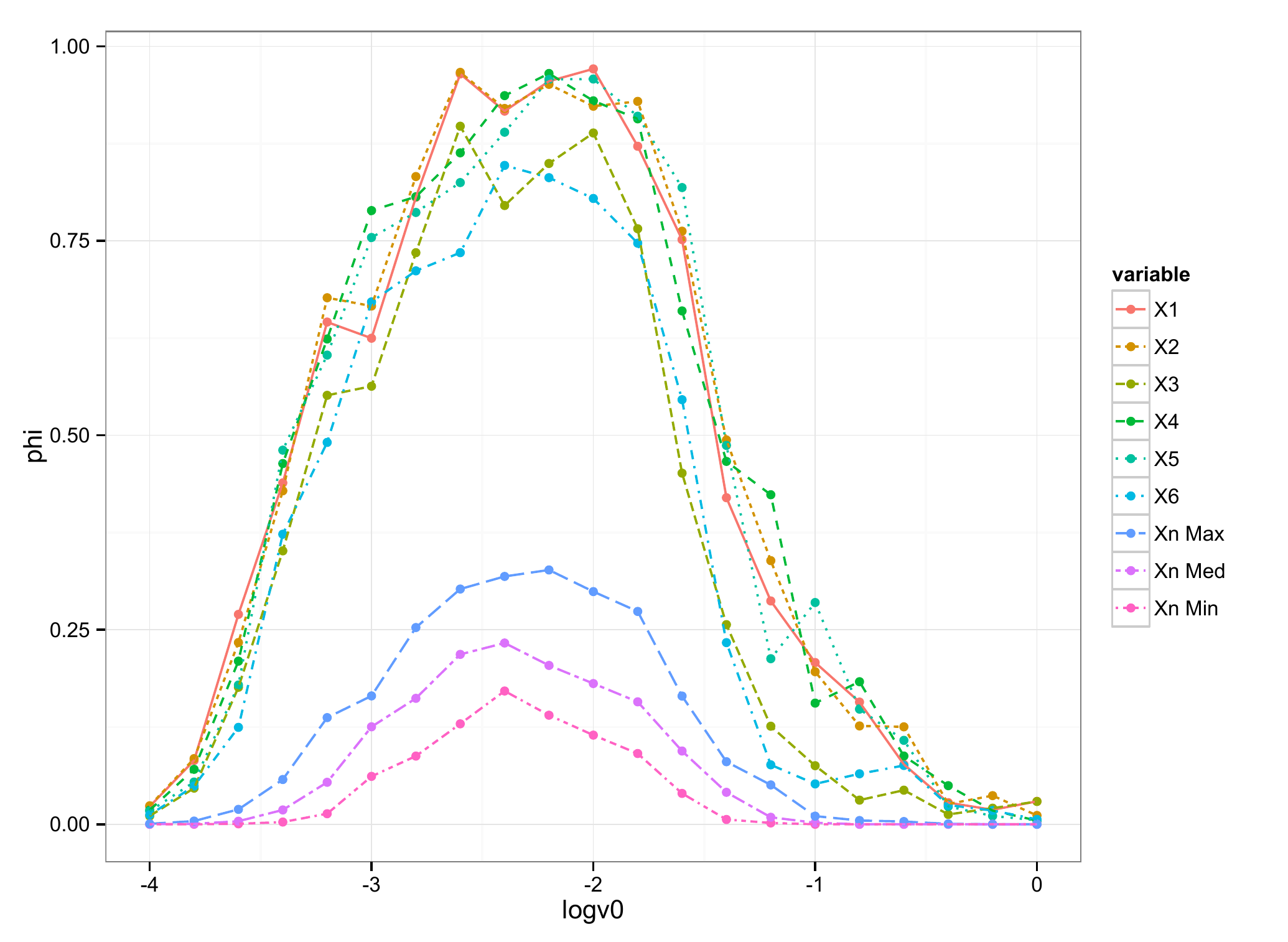}}
 \caption{Highly-correlated data $n=50$. A path-plot of the average selection frequency when  $v_0$ varies in the logarithm scale of base 10. Top 6 lines represent the true variables $\x_{1:6}$ and the bottom 3 lines represent the maximum, median and minimum among the noise variables $\x_{7:40}$.}
 \label{fig:tuningcorrelated}
\end{figure} 

\subsection{A Large-$p$ small-$n$ example}
Finally we apply BBEM  on a large-$p$ small-$n$ example from \cite{George:2014}, where $p=1000$ and $n=100$. Each of the $p$ features  is generated from a standard normal with pairwise correlation to be $0.6^{|i-j|}$ and the response $\y$ is generated from the following linear model:
$$\y = \x_{1} + 2\x_2 +3\x_3 + \epsilon, $$
where $\epsilon \sim \sfN(0, 3).$

For this large $p$ example, we set the parameters in the BBEM algorithm as follows: the initial value of $\theta$ is $\sqrt{n}/p$, the number of variables used in each bootstrap iteration $L=n/2=50$ and the total number of replicates $K=100$. It is well known that cross-validation based on prediction accuracy tends to include more noise variables. So, for this example where the true model is known to be sparse, we choose to tune $v_0$ via BIC. For illustration purpose, we also include BBEM with a fixed  tuning parameter $v_0=0.03$ in the comparison group.  We compare BBEM with the EMVS algorithm from \cite{George:2014}, which is implemented by us using the annealing technique for $\bb$'s initialization, and fixed $v_0=0.5, v_1=1000$ as suggested in \cite{George:2014}. 
 
Table ~\ref{table1} reports the average number of signal and noise variables being selected over 100 iterations for each method. BBEM with BIC tuning performs the best: it selects $2.99$ signal variables out of $3$ on average (i.e.,  only miss one variable, the weakest signal $\x_1$, once in all $100$ iterations) and meanwhile has the smallest type I error. The BBEM algorithm with a fixed tuning parameter has a similar result as EMVS but is much faster.  The computation advantage for BBEM comes from two aspects: the computation trick that reduces the computation cost on matrix inversion and the sub-sampling step in Bayesian bootstrap which allows us to deal with just a subset of variables of size smaller than $p$. 
\begin{table}[h]
\centering
\begin{tabular}{l|rrr}
\hline
      & $x_j \in$ Signal & $x_j \in$ Noise \\
\hline
BBEM (BIC)  &   2.99       &      0.24        \\
BBEM  ($v_0 = 0.03$)&   2.96       &      0.27      \\
EMVS  &   2.97       &      0.29       \\
Oracle & 3 & 0 \\
\hline
\end{tabular}
\label{table1}
\caption{A large-$p$ small-$n$ example. The table shows the average number of signal and noise variables being selected out of 100 iterations. In BBEM, $v_0$  is either chosen by BIC or fixed at $0.03$. EMVS is the algorithm proposed by \cite{George:2014}. }
\end{table}

\subsection{A real example}

For TFI, a company that owns some of the world's most well-known brands like Burger King and Arby's, decisions on where to open new restaurants are crucial. It usually takes a big investment of both time and capital at the beginning to set up a new restaurant. If a wrong location is chosen, likely the restaurant will soon be closed and all the initial investment will be lost. TFI hosted a  prediction competition on Kaggle\footnote{\url{https://www.kaggle.com/c/restaurant-revenue-prediction}}, where the goal is to build a mathematical model to predict the revenue of a restaurant based on a set of 
demographic, real estate, and commercial information. The data contains 137 restaurants in the training set and 1000 restaurants in the test set. Features include the Open Date, City, City Group, Restaurant Type, and three categories of obfuscated data (P1-P37, numeric): demographic data, real estate data, and commercial data. The response  is the transformed restaurant revenue in a given year. 

We first transform the ``Open Date" to a numeric feature called ``Year Since 1900" and merge the ``City" column into the ``City Group" column which now contains four categories: Istanbul, Izmir, Ankara, and others (small cities). 
Then we crate dummy variables for the categorical features like ``City Group"  and ``Restaurant Type" and keep all the obfuscated numeric columns P1-P37. The final training set has 43 features and 137 samples. 

After standardizing the data,
we fix $v_1$ at 100 and tune $v_0$ from $10^{-4.5}$ to $10^{-0.5}$ for the BBEM algorithm, where each bootstrap sample uses $L=15$ variables, and the total number of replicates is $K=300$. The path-plot of selection frequency for important features is shown in Figure~\ref{fig:tuningRestaurant}.  It is not surprising that ``City Group", ``Years Since 1900" and ``Restaurant Type" are  important predictors for the revenue. Quite a few obfuscated features are also selected as important predictors. Although we do not know their meanings, they should provide valuable information for TFI to choose  their next restaurant's location. 
\begin{figure}[ht]
 \centering 
 \scalebox{0.55} 
 {\includegraphics{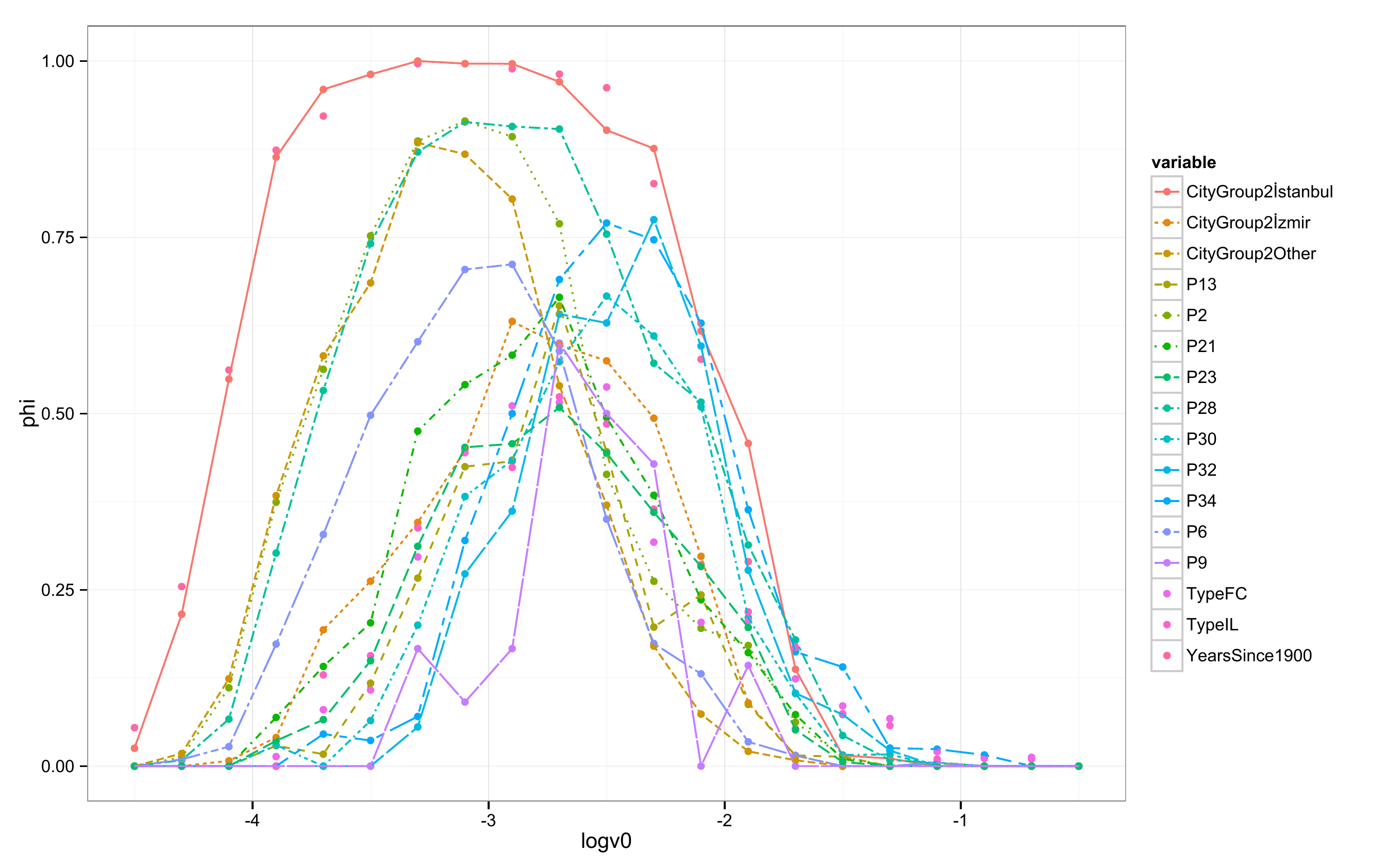}}
 \label{fig:tuningRestaurant}
 \caption{Restaurant data. The path plot of selection frequency when $v_0$ varies in the logarithm scale of base 10. Only a subset of variables with high selection frequencies are displayed. }
 \end{figure}

 Since the evaluation metric for this specific  competition is based on the rooted mean square error (RMSE), we use the same metric in our 5-fold cross-validation. 
We tuned $v_0$ from the set $\{0.0001, 0.0002, 0.0005, 0.001, 0.002, 0.005, 0.01\}$, and found $v_0=0.002$ has the smallest RMSE score. Then we fix $v_0$ at $0.002$, and re-run BBEM  on the whole training data. Let $\m$ denote the averaged posterior mean of $\bb$ from $L$ bootstrap iterations, and $\gg$ the averaged selection frequency for $p$ variables. We then use $\m*\gg$ (where $*$ denotes element-wise product) for prediction in the same spirit as the Bayesian model averaging. Our final Kaggle score is 1989762.52, which outperforms the random forest benchmark (RMSE=1998014.94) provided by Kaggle\footnote{At Kaggle, each team can submit their prediction and see the corresponding performance on the test data many times, so one can easily obtain a good score by keep tweaking the model to overfit the test data. For this reason, we did not compare our result with those ``low" scores on the leaderboard provided by individual teams.}. 
It is impressive for BBEM to outperform random forest considering that BBEM does not use any nonlinear features but random forest does.

\section{Further Discussion}
Variable selection is an important problem in modern statistics. In this paper, we study the Bayesian approach to variable selection in the context of multiple linear regression. We proposed an EM algorithm that returns the MAP estimate of the set of relevant variables. The algorithm can be operated very efficiently and therefore can scale up with big data. In addition, we have shown that the MAP estimate from our algorithm provides a consistent estimator of the true variable set even when the model dimension  diverges with the sample size. Further, we propose an ensemble version of our EM algorithm based on Bayesian bootstrap, which, as demonstrated via real and simulated examples, can substantially increase accuracy while maintaining the computation efficiency. 

Although we restrict our discussion for the linear model, the two algorithm we proposed can be easily extended to other generalized linear models by using latent variables \citep{Polson13PG}, an interesting topic for future research.

\appendix

\section*{Appendix: Proof of theorem \ref{thm:consistency}}\label{app:proof}
\begin{proof}
Recall the EM algorithm returns
\begin{equation*}
\hat{\gamma}_{nj} =  1, \quad\text{if}\quad \EE_{\bb|\Theta^{(t)}, \y} \big [ \beta_j^2 \big ] > r_n,
\end{equation*}
where the threshold
\[ r_n = \frac{\hat{\sigma}_n^2}{1/v_0 - 1/v_1} \Big  (\log{v_1\over v_0} -2 \log\frac{\hat{\theta}_n}{1-\hat{\theta}_n} \Big )=O(n^{-r_0}\log n) \]
and the conditional second moment of $\beta_j$ is equal to $m_j^2+\hat{\sigma}^2_n V_{jj}$ with 
 \begin{eqnarray*}
 \m &=& (\X_n^T \X_n  + D^{-1})^{-1} \X_n^T ( \X_n\bb_n^\ast+\e_n)\\
 &=& \bb^\ast- (\X_n^T \X_n  + D^{-1} )^{-1} D^{-1} \bb_n^\ast + (\X_n^T \X_n  + D^{-1})^{-1} \X_n^T\e_n\\
 &=&\bb^\ast-\b_n+\mathbf{W}_n\\
 \V &=& (\X_n^T \X_n +D^{-1})^{-1}, \quad D^{-1} =
\diag\left(\frac{1-\hat{\gamma}_{nj}}{v_0}+\frac{\hat{\gamma}_{nj}}{v_1}\right). 
 \end{eqnarray*}
 Here we represent the posterior mean of $\bb$ as three separate terms: the true coefficient vector $\bb^\ast_n$, the bias term $\b_n$ and the random error term $W_n$. So the event $\{ \hat{\gg}_n = \gg_n^* \} $ is equivalent to
\begin{equation}\label{th1:twosides}
\left\{ \min_{j: \gamma_{nj}^* = 1} m_j^2 + \hat{\sigma}^2_n V_{jj} > r_n \right \} \cap \left \{  \max_{j: \gamma_{nj}^* =0}  m_j^2 + \hat{\sigma}^2_n V_{jj} < r_n \right \}. 
\end{equation} 

First we prove the following results that quantify $m_j^2$ and $V_{jj}$. 
\begin{itemize}
\item[(R1)] $V_{jj}$ is upper bounded by the largest eigenvalue of $\V$, 
\begin{equation}
V_{jj} \le \frac{1}{\lambda_{n1}+1/v_1} = O(n^{-\eta_1})\prec O(n^{-r_0}\log n)=r_n,
\end{equation}
where for two sequences $\{a_n\}$ and $\{b_n\}$, we write $a_n \prec b_n$  if $a_n/b_n \to 0.$

\item[(R2)] The bias term $\b_n$ is bounded by
\begin{eqnarray}
\max_j |b_{nj}|\le \| \b_n\|_2 &\le & \|(\X_n^T \X_n  + D^{-1})^{-1}\|_2 \cdot \| D^{-1} \bb_n^\ast\|_2 \nonumber\\
&\le & \frac{1/v_0}{\lambda_{n1}+1/v_1}\|\bb_n^\ast\|_2 = O(n^{r_0-\eta_1+\eta_2}). \label{bound:bias}
\end{eqnarray}
When $r_0 <2(\eta_1-\eta_2)/3$, $\max_{j}  | b_{nj}|^2  \prec O(n^{-r_0}\log n)=r_n$.

The matrix $L_2$ norm is defined as $\| A\|_2 = \sup_{\|v\|=1} \|Av\|_2$, which is equal to its largest eigenvalue (singular value) when $A$ is symmetric (non-symmetric).  

\item[(R3)] Note that $\mathbf{W}_n$ is not a Gaussian random vector due to the dependence between $D$ and $\e_n$, but it can be rewritten as
\begin{eqnarray*}
\mathbf{W}_n &=& (\X_n^T \X_n  + D^{-1})^{-1} (\X_n^T\X_n)(\X_n^T \X_n)^{-1}  \X_n^T\e_n= A \tilde{\mathbf{W}}_n.
\end{eqnarray*}
where $A = \left (\X_n^T \X_n + D^{-1} \right )^{-1} \left (\X_n^T \X_n \right)$ and $\tilde{W}_n=(\X_n^T \X_n)^{-1}  \X_n^T\e_n$. Since $A$ is a matrix with norm bounded by $1$, we have
$$
\max_j | W_{nj}|  \le  \| A \|_{\infty}  \max_j | \tilde{W}_{nj}| \le \sqrt{p} \|A \|_2\max_j | \tilde{W}_{nj}| \le  \sqrt{p} \max_j | \tilde{W}_{nj}| .
$$

\item[(R4)]  $\tilde{\mathbf{W}}_n=(\X_n^T \X_n)^{-1}  \X_n^T\e_n$ is a Gaussian random vector with covariance $\sigma^2(\X_n^T \X_n)^{-1}$ and mean $\mathbf{0}$. So the variance for $W_{nj}$ is upper bounded by $ \sigma^2 \lambda_{n1}^{-1}$. 

Recall the tail bound for Gaussian variables: for any $Z \sim \textsf{N}(0, \tau^2)$,
 \[ \mathbb{P} ( | Z |  > t) = \mathbb{P} ( |Z|/ \tau > t/\tau) \le  \frac{\tau}{t} e^{- \frac{t^2}{2\tau^2} }. \]

With Result (R3) and Bonferroni's inequality, we can find a constant $M>0$ such that
\begin{eqnarray*}
\mathbb{P}  ( \max_{j}| W_{nj} | >  \sqrt{r_n} ) & \le &  \mathbb{P}  ( \max_{j}| \tilde{W}_{nj} | >  \sqrt{r_n/p}  ) \nonumber \\
 & \le &  p\cdot \mathbb{P}  (| \tilde{W}_{nj} | >  \sqrt{r_n/p}  ) \nonumber \\
&\le & \frac{p\sqrt{p}\sigma}{\sqrt{r_n\lambda_{n1}}}e^{-\frac{r_n\lambda_{n1}}{2p\sigma^2}}=O\big (e^{- Mn^{\eta_1-r_0-\alpha}} \big ), \label{proof:asym:tailprob}
\end{eqnarray*}
which goes to $0$ when $r_0 < \eta_1-\alpha$. So with probability going to $1$, $\max_{j}| W_{nj} | $ is upper bounded by $\sqrt{r_n}.$

\item[(R5)] When $1-\eta_3 < r_0$, $\min_{j: \gamma_{nj}^\ast =1} |\beta^{\ast}_{nj}|^2\sim n^{\eta_3-1}\succ O(n^{-r_0}\log n)=r_n$.
\end{itemize}

Now we prove (\ref{th1:twosides}). Given $1-\eta_3 < r_0 < \min\{\eta_1-\alpha, 2(\eta_1-\eta_2)/3\}$, we have
\begin{eqnarray*}
\mathbb{P}\left(\max_{j: \gamma_{nj}^\ast =0}  (m_j^2 + \hat{\sigma}_n^2 V_{jj} ) >r_n\right)
&\le & \mathbb{P}\left(\big( \max_{j}  | b_{nj}| + \max_{j}  |W_{nj}| \big)^2 + \hat{\sigma}_n^2 \max_j V_{jj}>r_n \right)\\
&\le & \mathbb{P}\left(\max_{j}  |W_{nj}| >\sqrt{r_n}\right)= O\big (e^{- Mn^{\eta_1-r_0-\alpha}} \big ), \\
\mathbb{P}\left(\min_{j: \gamma_{nj}^\ast =1}  (m_j^2 + \hat{\sigma}_n^2 V_{jj}  ) <r_n\right)
&\le & \mathbb{P}\left(\min_{j: \gamma_{nj}^\ast =1}|\beta^{\ast}_{nj}|^2- \big(\max_{j}  | b_{nj}|+ \max_{j}  |W_{nj}|  \big )^2<r_n \right)\\
&\le & \mathbb{P}\left(\max_{j}  |W_{nj}| >\sqrt{r_n}\right)= O\big (e^{- Mn^{\eta_1-r_0-\alpha}} \big ). 
\end{eqnarray*}
So (\ref{th1:twosides}) holds with probability $1-O(e^{- Mn^{\eta_1-r_0-\alpha}} )\to 1$. 
\end{proof}

\bibliographystyle{asa}
\bibliography{bib/BBEM}


\end{document}